\documentclass[11pt]{article}
\usepackage{mathrsfs}
\usepackage{CJK}
\usepackage{bbm}
\usepackage{graphicx}
\usepackage{amsmath,amssymb}
\usepackage{amsfonts}
\usepackage{amsthm}
\usepackage{verbatim}
\usepackage{color}

\usepackage[numbers,sort&compress]{natbib}

\newtheorem{Definition}{Definition}[section]
\newtheorem{Theorem}{Theorem}[section]
\newtheorem{Remark}{Remark}[section]
\newtheorem{Lemma}{Lemma}[section]

\numberwithin{equation}{section}

\makeatletter
\renewcommand\@biblabel[1]{#1.}
\makeatother

\begin{document}
\title{{\LARGE \textbf{The high order block RIP condition for signal recovery  }}}
\author{Chen Wengu$^{1}$ ,\ \  Li Yaling$^{2,}$\thanks{Corresponding author}\\[5pt]
$^{1}$ Institute of Applied Physics and Computational Mathematics\\ Beijing, 100088, China\\
$^{2}$ Graduate School, China Academy of Engineering Physics\\ Beijing, 100088, China\\[5pt]
Email: chenwg@iapcm.ac.cn, leeyaling@126.com} \maketitle

\begin{bfseries}
Abstract
\end{bfseries}
In this paper, we consider the recovery of block sparse signals,
 whose nonzero entries appear in blocks (or clusters) rather than spread arbitrarily throughout the signal, from incomplete linear measurement.
A high order sufficient condition based on block RIP is obtained to guarantee the stable recovery of all block sparse signals in the presence of noise, and robust recovery when signals are not exactly block sparse via mixed $l_{2}/l_{1}$ minimization. Moreover, a concrete example is established to
ensure the condition is sharp. The significance of the results presented
in this paper lies in the fact that recovery may be possible under more general conditions
 by exploiting the block structure of the sparsity pattern
  instead of the conventional sparsity pattern.

\begin{bfseries}
Keywords
\end{bfseries}
block sparsity $\cdot$  block restricted isometry property $\cdot$  compressed sensing $\cdot$  mixed $l_{2}/l_{1}$ minimization

\begin{bfseries}
Mathematics Subject Classification (2010)
\end{bfseries}
90C59 
$\cdot$
94A12   
$\cdot$
94A20   

\section{Introduction}\label{intro}
\quad\quad  Compressed sensing (CS), a new type of sampling theory, is a fast growing field of research. It has attracted considerable interest in a number of fields including
applied mathematics, statistics, seismology, signal processing and electrical engineering. Interesting applications include radar system \cite{HS,ZZZ}, coding theory
\cite{AT,CT}, DNA microarrays\cite{PVMH},
color imaging \cite{MW}, magnetic resonance imaging \cite{LDSP}.
Up to now, there are already many works on CS \cite{BDDW,BW,CDD,CDD1,CDD2,T,T1,LT,LT1,R, RRT,TG}.
The key problem in CS is to recover an unknown high-dimensional sparse signal $x\in \mathbb{R}^{N}$ using an efficient algorithm through a sensing matrix $A\in \mathbb{R}^{n\times N}$
and
the following linear measurement
\begin{align}\label{m1}
y=Ax+z
\end{align}
where observed signal $y\in \mathbb{R}^{n}$,
$n\ll N$ and $z \in \mathbb{R}^{n}$ is a vector of measurement errors.
In general, the solutions to the underdetermined system of linear equations (\ref{m1})
are not unique. But now suppose that $x$ is known to be sparse
in the sense that it contains only a small number of nonzero entries, which can occur in anywhere in $x$.
This premise fundamentally changes the problem such that there is a unique sparse solution under regularity conditions.
It is well known the $l_{1}$ minimization approach, a widely used algorithm,
is an effective way to recover
sparse signals in many setting. One of the most widely used frameworks
to depict recovery ability of $l_{1}$ minimization
in CS is the restricted isometry property (RIP) introduced by Cand\`es and Tao \cite{CT}.
Let $A\in \mathbb{R}^{n\times N}$ be a matrix and $1\leq k \leq N$
is an integer, the restricted isometry constant (RIC) $\delta_{k}$
of order $k$ is defined as the smallest nonnegative
 constant that
satisfies
$$(1-\delta_{k})\|x\|_{2}^{2}\leq\|Ax\|_{2}^{2}\leq(1+\delta_{k})\|x\|_{2}^{2},$$
for all $k-$sparse vectors $x\in\mathbb{R}^{N}.$ A vector
$x\in \mathbb{R}^{N}$ is $k-$sparse if $|\mathrm{supp}(x)|\leq k$, where
$\mathrm{supp}(x)=\{i: x_{i}\neq 0\}$ is the support of $x$. When $k$ is not an integer, we define $\delta_{k}$ as $\delta_{\lceil k \rceil}$.
It has been shown $l_{1}$ minimization can recover a sparse signal with
a small or zero error under some appropriate RIC met by the measurement matrix $A$ \cite{CZ2,CWX1,CXZ,CWX,CZ1,ML,CZ,CRT,C1,F,F1}.
As far as we know,  a sharp sufficient condition based on RIP for exact and stable recovery of signals
in both noiseless and noisy cases by $l_{1}$ minimization was established by Cai and Zhang \cite{CZ}.

However, in practical examples, there are signals which have a particular sparsity pattern, where the nonzero coefficients appear in some blocks (or clusters). Such signals are referred to as block sparse \cite{EM,EKB,WWX}.
In practice, the block sparse structure is very common, such as reconstruction of multi-band signals \cite{ME}, equalization of sparse communication channels \cite{CR} and multiple measurement vector (MMV) model \cite{EM,ME1,ER}.
Actually, the notion of block sparsity was already introduced in statistics literature and was named the group Lasso estimator
\cite{YL,CH,NR,B,MGB,HZ}. Recently, block sparsity pattern has attracted significant attention. Various efficient methods
and explicit recovery guarantees \cite{LL,FLZZ,LWB,WBWL,EV,WZLLT,WWX,EKB,EM} have been proposed.

In this paper, our goal is to  recover the unknown signal $x$ from linear measurement \eqref{m1}. But at the moment,  nonzero elements of signal $x$
are occurring in blocks (or clusters) instead of spreading arbitrarily throughout the signal vector. To this end, firstly, we need the concept of block sparsity. In order to emphasize the block structure, similar to \cite{WWX,EM}, we view $x$ as a concatenation of blocks over $\mathcal{I}=\{d_{1},d_{2},\ldots, d_{M}\}$. Then $x$ can be expressed as
$$x=(\underbrace{x_{1},\ldots,x_{d_{1}}}_{x[1]}, \underbrace{x_{d_{1}+1},\ldots,x_{d_{1}+d_{2}}}_{x[2]},
\ldots,\underbrace{x_{N-d_{M}+1},\ldots, x_{N}}_{x[M]})^{T} \in \mathbb{R}^{N},$$
where $x[i]$ denotes the $i$th block of $x$ with the length $d_{i}$ and $N=\sum_{i=1}^{M}d_{i}$. A vector $x\in\mathbb{R}^{N}$ is called block $k-$sparse over $\mathcal{I}=\{d_{1},d_{2},\ldots, d_{M}\}$
if the number of nonzero vectors $x[i]$ is at most $k$ for $i\in \{1,2,\ldots,M\}$. Define
$$\|x\|_{2,0}=\sum_{i=1}^{M}I(\|x[i]\|_{2}>0),$$
where $I(\cdot)$ is an indicator function that it equals to 1 if
its argument is larger than zero and $0$ elsewhere. Then the block $k-$sparse vector over $\mathcal{I}=\{d_{1},d_{2},\ldots, d_{M}\}$
can be cast as $\|x\|_{2,0}\leq k$.
If $d_{i}=1$ for all $i$,
block sparsity is just the conventional sparsity.
Next, one of the efficient methods to recover block sparse signals
is mixed $l_{2}/l_{1}$ minimization
\begin{align}\label{g11}
 \min_{x} \|x\|_{2,1}, \quad \|y-Ax\|_{2}\leq \varepsilon,
\end{align}
where $\|x\|_{2,1}=\sum_{i=1}^{M}\|x[i]\|_{2}$.
Moreover, mixed norm
$\|x\|_{2,2}=(\sum_{i=1}^{M}\|x[i]\|_{2}^{2})^{1/2}$ and
$\|x\|_{2,\infty}=\max_{i}\|x[i]\|_{2}$. Note that $\|x\|_{2,2}=\|x\|_{2}$.
It is easy to know the mixed norm minimization is a generalization of conventional norm minimization.
To ensure uniqueness and stability of solution for the system \eqref{m1}
via mixed $l_{2}/\l_{1}$ minimization,
Eldar and Mishali \cite{EM} generalized the notion of standard restricted isometry property to
block sparse vectors, and obtained the following concept of
block restricted isometry property (block RIP).

\begin{Definition}[block RIP]
Let $A\in \mathbb{R}^{n\times N}$ be a matrix,
then $A$ has the $k$ order block restricted isometry property over $\mathcal{I}=\{d_{1},d_{2},\ldots, d_{M}\}$ with nonnegative parameter $\delta_{k|\mathcal{I}}$ if
$$(1-\delta_{k|\mathcal{I} })\|x\|_{2}^{2}\leq\|Ax\|_{2}^{2}\leq(1+\delta_{k|\mathcal{I}})\|x\|_{2}^{2}$$
holds for all block $k-$sparse vector $x\in\mathbb{R}^{N}$ over $\mathcal{I}$.
The smallest constant $\delta_{k|\mathcal{I}}$ is called block restricted isometry
constant (block RIC).
When $k$ is not an integer, we define $\delta_{k|\mathcal{I} }$ as $\delta_{\lceil k\rceil|\mathcal{I}}$.
\end{Definition}
For simplicity, we use $\delta_{k}$ for the block RIP
 constant $\delta_{k|\mathcal{I}}$ in the remainder of this paper.
The block RIP plays a role similar to
standard RIP. The block RIP provides recovery guarantee for block sparse signals. For example, Eldar and Mishali \cite{EM} proved that if matrix $A$ satisfies block restricted isometry constant (block RIC) $\delta_{2k}<\sqrt{2}-1$,   the mixed $l_{2}/\l_{1}$ minimization
can recover exactly the block $k-$sparse signals in noiseless case,
and can approximate the best block $k-$sparse solution
in the presence of noise and mismodeling errors.
Furthermore, they illustrated the advantage of block RIP over standard
RIP. That is, the probability to satisfy the standard RIP
is less than that of satisfying the block RIP.
Meanwhile, a specific example is given to account for
the advantage. They also experimentally demonstrated
the advantage of their algorithm (mixed $l_{2}/l_{1}$ minimization)
over standard basis pursuit. This explained the performance advantage of
block sparse recovery over standard sparse recovery.
Later, Lin and Li \cite{LL} improved the bound of block RIC
to $\delta_{2k}<0.4931$, and also gave another one order
sufficient condition of recovery based on block RIC $\delta_{k}<0.307$.
There are many other recovery guarantees and efficient recovery methods
to  ensure the recovery of signals with special structure. For example, block coherence \cite{EKB}, strong group sparsity \cite{HZ}, $l_{2}/l_{p}(0<p<1)$ minimization \cite{WWX}, BOMP\cite{FLZZ,WZLLT}.

In this paper, we investigate the high order block RIP conditions for the exact or stable recovery of signals with blocks structure from \eqref{m1} via solving mixed $l_{2}/l_{1}$ minimization in noiseless and noise case.
 Using ideas similar to \cite{CZ},
we establish a sufficient condition on $\delta_{tk}$ to ensure the
stable or exact recovery of signals with nonzero entries occurring
in blocks (or clusters) rather than being arbitrarily spread
throughout the signal vector. The key is to generalize the technique
of sparse representation of a polytope \cite{CZ} to the block
setting. We show that block RIC $\delta_{tk}<\sqrt{\frac{t-1}{t}}$
for any $t>1$ can ensure exact and stable recovery for all block
sparse signals and robust recovery for nearly block sparse signals
via mixed $l_{2}/l_{1}$ minimization. Moreover, it is sharp when
$t\geq 4/3$. A concrete example is given to  illustrate the
optimality. Actually, our results are a generalization of that of
Cai and Zhang \cite{CZ} in the block setting. When $d_{i}=1$ for
$i\in\{1,\ldots,M\}$, our results return to those of Cai and Zhang
\cite{CZ}. The significance of our results lies in the fact that
taking advantage of explicit block sparsity has better
reconstruction performance than viewing the signals as being
standard sparsity, accordingly ignoring the additional structure in
the problem.

The rest of the paper is organized as follows. In Section \ref{2},
we will introduce some notations and establish some basic lemmas that will be used.
The main results and their proofs are given in Section \ref{3}. Finally,
we summarize this paper in Section \ref{4}.

\section{Preliminaries}\label{2}
Throughout this paper, we adopt the following notations unless otherwise stated.
For any $x\in\mathbb{R}^{N}$,
we model it over $\mathcal{I}=\{d_{1},d_{2},\ldots, d_{M}\}$. $
x[i]$ denotes the $i$th block
of $x$. Let $\textbf{0}$ be the zero vector
 whose dimension may be different.
$\Gamma\subset\{1, 2, \ldots, M\}$ indicates block indices, $\Gamma^{c}$ is the complement of $\Gamma$ in $\{1, 2, \ldots, M\}$.  $x[\Gamma]\in \mathbb{R}^{N}$ denotes a vector which equals to $x$ on
block indices $\Gamma$ and $0$ otherwise.
For example, if $\Gamma=\{1, 3, M\}$, then
$x[\Gamma]=(x[1], \textbf{0}, x[3], \textbf{0},\ldots,\textbf{0},
x[M])^{T} \in \mathbb{R}^{N},$
and $x[\Gamma][i]$ denotes
$i$th block
of $x[\Gamma]$.
We denote by  $\mathrm{supp}[x]=\{i: \|x[i]\|_{2}\neq 0\}$ the block support of $x$, and $\mathcal{I}_{0}$ the block indices of the $k$ largest block in $l_{2}$ norm of $x$, i.e., $\|x[i]\|_{2}\geq \|x[j]\|_{2}$ for any $i\in \mathcal{I}_{0}$ and $j\in \mathcal{I}_{0}^{c}$.
We also denote $x[max(k)]$ as $x$ with all but the largest
$k$ blocks in $l_{2}$ norm set to zero.
From now on, we always take that $h=\widehat{x}-x$, where $\widehat{x}$ is the minimizer of $l_{2}/l_{1}$ minimization problem \eqref{g11} and $x$ is the original signal.

The following lemma provides a key technical tool for the proof of our main result. It is an extension of Lemma 1.1 introduced by  Cai and Zhang \cite{CZ}.
We extend sparse representation of a polytope to the block setting.
\begin{Lemma}\label{l1}
For a positive number $\alpha$ and a positive integer $k$,
define the block polytope $T(\alpha, k)\subset\mathbb{R}^{N}$ by
$$T(\alpha, k)=\{v\in\mathbb{R}^{N}: \|v\|_{2,\infty}\leq \alpha, \|v\|_{2,1}\leq k\alpha\}.$$
For any $v\in \mathbb{R}^{N}$,
define the set of block sparse vectors $U(\alpha, k, v)\subset\mathbb{R}^{N}$ by
$$U(\alpha, k, v)=\{u\in\mathbb{R}^{N}: \mathrm{supp}(u)\subseteq \mathrm{supp}(v), \|u\|_{2,0}\leq k, \|u\|_{2,1}=\|v\|_{2,1}, \|u\|_{2,\infty}\leq\alpha\}.$$
Then any $v\in T(\alpha, k)$ can be expressed as
$$v=\sum\limits_{i=1}^{J}\lambda_{i}u_{i},$$
where $u_{i}\in U(\alpha, k, v)$ and $0\leq \lambda_{i}\leq 1, \sum\limits_{i=1}^{J}\lambda_{i}=1.$
\end{Lemma}

\begin{proof}
First of all, what we have to prove is that $v\in T(\alpha, k)$ is in the convex hull of $U(\alpha, k, v)$. To show the statement, we proceed by induction.

Suppose $v\in T(\alpha, k)$.
If $v$ is block $k-$sparse, $v$ itself is in $U(\alpha, k, v)$.
Thus, assume that the assertion is true for all block $(s-1)-$sparse vectors $v$ $(s-1\geq k)$, then we show that the assertion is also true for any block $s-$sparse vectors $v$. For any block $s-$sparse vectors $v \in T(\alpha, k)$, (without loss of generality, suppose that $v$ is not block $(s-1)-$sparse, otherwise the result holds by assumption of block $(s-1)-$sparse), we have $\|v\|_{2,\infty}\leq \alpha, \|v\|_{2,1}\leq k\alpha$. Furthermore, $v$ can be expressed as  $v=\sum\limits_{i=1}^{s}c_{i}E_{i}$ with $c_{1}\geq c_{2}\cdots \geq c_{s}>0$, where $c_{1}$ equals to the largest $\|v[i]\|_{2}$ for every $i \in \{1,2,\ldots, M\}$, $c_{2}$ equals to the next largest $\|v[i]\|_{2}$, and so on, where $E_{i}$ is a unit vector in $\mathbb{R}^{N}$, which equals to $v/c_{i}$ on the $i$th largest block of $v$ and zero elsewhere.
Let $$\Omega=\{1\leq l\leq s-1: c_{l}+c_{l+1}+\cdots+c_{s}\leq(s-l)\alpha\}.$$
Owing to $\sum\limits_{i=1}^{s}c_{i}=\|v\|_{2,1}\leq k\alpha$, $\Omega$ is not empty for $1\in \Omega$. We denote by $l$ the largest element in $\Omega$. It is easy to get
\begin{align}\label{g1}
  c_{l}&+c_{l+1}+\cdots+c_{s}\leq(s-l)\alpha,\notag \\
  c_{l+1}&+c_{l+2}+\cdots+c_{s}>(s-l-1)\alpha.
\end{align}
It is worthy of noting that \eqref{g1} also holds when the largest element in $\Omega$ is $s-1$.
Take \begin{align}\label{g2}
      b_{j}=\frac{\sum_{i=l}^{s}c_{i}}{s-l}-c_{j}, \quad l\leq j\leq s.
     \end{align}
By direct calculations, we have  $(s-l)\sum_{i=l}^{s}b_{i}=\sum_{i=l}^{s}c_{i}$ and $b_{j}\geq b_{l}$ for all $l\leq j\leq s$. Moreover, for any $l\leq j\leq s$,
\begin{align*}
 b_{j}&\geq b_{l}=\frac{\sum_{i=l+1}^{s}c_{i}-(s-l-1)c_{l}}{s-l}\\
 &\geq \frac{\sum_{i=l+1}^{s}c_{i}-(s-l-1)\alpha}{s-l}>0,                           \end{align*}
 where the second inequality follows from the fact that $\|v\|_{2,\infty}\leq \alpha$, the last inequality is a result of the second inequality in \eqref{g1}.
 Next, define
 \begin{align*}
   \lambda_{j}  =\frac{b_{j}}{\sum_{i=l}^{s}b_{i}},  \quad  v_{j} =\sum_{i=1}^{l-1}c_{i}E_{i}+(\sum_{i=l}^{s}b_{i})\sum_{i=l,i\neq j}^{s}E_{i} \in \mathbb{R}^{N},\\
   \end{align*}
 for every $l\leq j\leq s$, we have  $v=\sum_{j=l}^{s}\lambda_{j}v_{j}$, $0<\lambda_{j}\leq 1$, $\sum_{j=l}^{s}\lambda_{j}=1$, $\mathrm{supp} (v_{j})\subseteq \mathrm{supp} (v)$. In addition, using the fact that
 $(s-l)\sum_{i=l}^{s}b_{i}=\sum_{i=l}^{s}c_{i}$, $\|v\|_{2,\infty}\leq \alpha$
 and the first inequality in \eqref{g1},
 \begin{align*}
   \|v_{j}\|_{2,1}&= \sum_{i=1}^{l-1}c_{i} + (s - l)\sum_{i=l}^{s}b_{i}
   =\sum_{i=1}^{l-1}c_{i} + \sum_{i=l}^{s}c_{i}=\|v\|_{2,1},\\
  \|v_{j}\|_{2,\infty}&=\max\{c_{1},\ldots,c_{l-1}, \sum_{i=l}^{s}b_{i}\}\leq \max \{\alpha, \frac{\sum_{i=l}^{s}c_{i}}{s-l}\}\leq \alpha.
 \end{align*}
 Finally, since $v_{j}$ is block $(s-1)-$sparse, under the induction assumption, we have  $v_{j}=\sum_{i=1}^{J}\mu_{j,i}u_{j,i}$, where $u_{j,i}$ is block $k-$sparse, $\|u_{j,i}\|_{2,1}=\|v_{j}\|_{2,1}=\|v\|_{2,1}$, $\|u_{j,i}\|_{2,\infty}\leq \alpha$ and $0\leq \mu_{j,i} \leq 1$, $\sum_{i=1}^{J}\mu_{j,i}=1$. Hence, $v=\sum_{i=1}^{J}\sum_{j=l}^{s}\lambda_{j}\mu_{j,i}u_{j,i}$, which implies that the statement is true for $s$.

 On the other hand, if $v$ is in the convex hull of $U(\alpha, k,v)$, then
 $v=\sum_{i=1}^{J}\lambda_{i}u_{i}$, $u_{i}\in U(\alpha, k,v)$ and $0\leq \lambda_{i}\leq 1$ , $\sum_{i=1}^{J}\lambda_{i}=1$. It follows immediately that
 \begin{align*}
   \|v\|_{2,1}&=\|\sum_{i=1}^{J}\lambda_{i}u_{i}\|_{2,1}\leq \sum_{i=1}^{J}\lambda_{i}\|u_{i}\|_{2,1}\leq \sum_{i=1}^{J}\lambda_{i}\|u_{i}\|_{2,0}\|u_{i}\|_{2,\infty}\leq k\alpha, \\
   \|v\|_{2,\infty}&=\|\sum_{i=1}^{J}\lambda_{i}u_{i}\|_{2,\infty}\leq \sum_{i=1}^{J}\lambda_{i}\|u_{i}\|_{2,\infty}\leq \alpha,
 \end{align*}
 which completes the proof.
\end{proof}

The following lemma is a useful elementary inequality,
which will be used in proving our main results.
\begin{Lemma}[\cite{CZ1}, Lemma 5.3]\label{l2}
Assume $m\geq k$, $a_{1}\geq a_{2}\geq\cdots\geq a_{m}\geq 0$,
$\sum\limits_{i=1}^{k}a_{i}\geq \sum\limits_{i=k+1}^{m}a_{i},$
then for all $\alpha\geq 1$,
$$\sum\limits_{j=k+1}^{m}a_{j}^{\alpha}\leq \sum\limits_{i=1}^{k}a_{i}^{\alpha}.$$
More generally, assume $a_{1}\geq a_{2}\geq\cdots\geq a_{m}\geq 0$, $\lambda\geq 0$
and $\sum\limits_{i=1}^{k}a_{i}+\lambda\geq \sum\limits_{i=k+1}^{m}a_{i},$ then for all $\alpha\geq 1$,
$$\sum\limits_{j=k+1}^{m}a_{j}^{\alpha}\leq k\Big(\sqrt[\alpha]{\frac{\sum_{i=1}^{k}a_{i}^{\alpha}}{k}}
+\frac{\lambda}{k}\Big)^{\alpha}.$$
\end{Lemma}
From the definition of $h$, $\widehat{x}$ and $x$, we have the following lemma.
\begin{Lemma}\label{l3}
For any $\Gamma\subset\{1, 2, \ldots, M\}$, it holds that
$$\|h[\Gamma^{c}]\|_{2,1}\leq \|h[\Gamma]\|_{2,1}+2\|x[\Gamma^{c}]\|_{2,1}.$$
\end{Lemma}
\begin{proof}
Recall that $h=\widehat{x}-x$. From the minimality of $\widehat{x}$,
it follows that
$\|\widehat{x}\|_{2,1}\leq\|x\|_{2,1}
=\|x[\Gamma]\|_{2,1}+\|x[\Gamma^{c}]\|_{2,1}$.
By the reverse triangle inequalities of $\|\cdot\|_{2,1}$, we obtain
\begin{align*}
  \|\widehat{x}\|_{2,1}&=\|x+h\|_{2,1}
=\|x[\Gamma]+h[\Gamma]\|_{2,1}+\|x[\Gamma^{c}]+h[\Gamma^{c}]\|_{2,1}\\
&\geq \|x[\Gamma]\|_{2,1}-\|h[\Gamma]\|_{2,1}+
\|h[\Gamma^{c}]\|_{2,1}-\|x[\Gamma^{c}]\|_{2,1}.
\end{align*}
The lemma follows from above inequalities immediately.
\end{proof}

\section{Main results}\label{3}
With the preparations given in Section \ref{2}, we establish the main results
in this section$-$a sharp high order block RIP condition for the robust recovery
of arbitrary signals with block pattern. When the signal is block sparse, the sharp condition ensures
 the exact recovery and stable recovery in the noiseless case and  in the noise case, respectively.
First, the following theorem provides a sufficient condition of recovery when $x$ is not block sparse and the observation is noisy.
\begin{Theorem}\label{t1}
Suppose that $x\in \mathbb{R}^{N}$ is an arbitrary vector consistent with \eqref{m1} and $\|z\|_{2}\leq \varepsilon$. If the measurement matrix $A$ satisfies the block RIP with
$\delta_{tk}<\sqrt{\frac{t-1}{t}}$ for $t> 1$,
the solution $\widehat{x}$ to \eqref{g11} obeys
\begin{align}\label{s7}
\|\widehat{x}-x\|_{2}
&\leq \frac{2\sqrt{2t(t-1)(1+\delta_{tk})}}{t(\sqrt{(t-1)/t}-\delta_{tk})}\varepsilon \notag\\
&+\Big(\frac{\sqrt{2}\delta_{tk}+\sqrt{t(\sqrt{(t-1)/t}-\delta_{tk})\delta_{tk}}}{t(\sqrt{(t-1)/t}-\delta_{tk})}+1\Big)
\frac{2\|x[\mathcal{I}_{0}^{c}]\|_{2,1}}{\sqrt{k}}.
 \end{align}

\end{Theorem}
\begin{proof}
First of all, suppose that $tk$ is an integer.
Let $\widehat{x}=x+h$, where $\widehat{x}$
is a solution to $l_{2}/l_{1}$ minimization problem and $x$ is the original signal.
From Lemma \ref{l3} and the definition of $\mathcal{I}_{0}$,
$\|h[\mathcal{I}_{0}^{c}]\|_{2,1}\leq \|h[\mathcal{I}_{0}]\|_{2,1}+2\|x[\mathcal{I}_{0}^{c}]\|_{2,1}.$
We assume that $T_{0}$ is the block index set
over the $k$ blocks with largest $l_{2}$ norm
of $h$.
Hence,
\begin{align}\label{g3}
  \|h[T_{0}^{c}]\|_{2,1}\leq \|h[T_{0}]\|_{2,1}+2\|x[\mathcal{I}_{0}^{c}]\|_{2,1}.
\end{align}
Denote
 $r=(\|h[T_{0}]\|_{2,1}+2\|x[\mathcal{I}_{0}^{c}]\|_{2,1})/k.$
 Next, we decompose $h[T^{c}_{0}]$ as
$h[T^{c}_{0}]=h[T_{1}]+h[T_{2}]$,
where $h[T_{1}]$ remains the blocks of $h[T^{c}_{0}]$
whose $l_{2}$ norms are greater than $\frac{r}{t-1}$ and $0$ elsewhere,
$h[T_{2}]$ retains the blocks of $h[T^{c}_{0}]$
whose $l_{2}$ norms are not more than $\frac{r}{t-1}$ and $0$ otherwise.
Combining above definitions and \eqref{g3}, we have
 $$\|h[T_{1}]\|_{2,1}\leq \|h[T^{c}_{0}]\|_{2,1}\leq kr.$$
Denote
$$\|h[T_{1}]\|_{2,0}=l.$$
From the definition of $h[T_{1}]$, we get
$kr\geq \|h[T_{1}]\|_{2,1}
=\sum\limits_{i\in \mathrm{supp}[h[T_{1}]]}||h[T_{1}][i]||_{2}
> \frac{lr}{t-1}.$
Namely, $l< k(t-1).$
Thus, it is clear that
\begin{align}\label{g4}
\|h[T_{2}]\|_{2,1}&=\|h[T^{c}_{0}]\|_{2,1}-\|h[T_{1}]\|_{2,1}   \notag\\
&< kr-\frac{lr}{t-1}    \notag\\
&=(k(t-1)-l)\cdot\frac{r}{t-1},    \notag\\
\|h[T_{2}]\|_{2,\infty}&\leq \frac{r}{t-1},
\end{align}
and
$\|h[T_{0}]+h[T_{1}]\|_{2,0}=k+l< tk.$
From the definition of $\delta_{k}$,
\begin{align}\label{g5}
 \langle A(h[T_{0}]+h[T_{1}]), Ah\rangle
  &\leq \|A(h[T_{0}]+h[T_{1}])\|_{2}\|Ah\|_{2} \notag\\
  &\leq \sqrt{1+\delta_{tk}}\|h[T_{0}]+h[T_{1}]\|_{2}\|Ah\|_{2}.
  \end{align}
Due to
\begin{align}\label{g6}
  \|Ah\|_{2}\leq \|A\widehat{x}-Ax\|_{2}\leq \|A\widehat{x}-y\|_{2}+\|y-Ax\|_{2}\leq 2\varepsilon,
\end{align}
so \eqref{g5} becomes
\begin{align}\label{g7}
\langle A(h[T_{0}]+h[T_{1}]), Ah\rangle
\leq \sqrt{1+\delta_{tk}}\|h[T_{0}]+h[T_{1}]\|_{2}\cdot (2\varepsilon).
\end{align}
Using \eqref{g4} and Lemma \ref{l1}, we have
$h[T_{2}]=\sum\limits_{i=1}^{J}\lambda_{i}u_{i},$
where $u_{i}$ is block $(k(t-1)-l)-$sparse,  $\sum\limits_{i=1}^{J}\lambda_{i}=1$ with $0\leq \lambda_{i}\leq 1$, and
$\mathrm{supp}(u_{i})\subseteq \mathrm{supp}(h[T_{2}])$,
$\|u_{i}\|_{2,1}=\|h[T_{2}]\|_{2,1}, \ \|u_{i}\|_{2,\infty}\leq \frac{r}{t-1}.$
Hence,
\begin{align*}
 \|u_{i}\|_{2}&=\|u_{i}\|_{2,2}\leq \sqrt{\|u_{i}\|_{2,0}}\|u_{i}\|_{2,\infty}\\
 &\leq \sqrt{k(t-1)-l}\cdot\frac{r}{t-1}
\leq \sqrt{\frac{k}{t-1}}r,
\end{align*}
where the first inequality follows from the fact that for any block $k-$sparse vector $ \upsilon$, $\|\upsilon\|_{2,2}^{2}=\sum_{i}\|\upsilon[i]\|_{2}^{2}\leq k\|\upsilon\|_{2,\infty}^{2}$.
Let
$X=\|h[T_{0}]+h[T_{1}]\|_{2,2}, \ P=\frac{2\|x[\mathcal{I}^{c}_{0}]\|_{2,1}}{\sqrt{k}}.$
Clearly,
\begin{align}\label{g8}
 \|u_{i}\|_{2}&\leq \sqrt{\frac{k}{t-1}}r \notag\\
  &=\sqrt{\frac{k}{t-1}}
  \cdot\frac{\|h[T_{0}]\|_{2,1}+2\|x[\mathcal{I}_{0}^{c}]\|_{2,1}}{k} \notag\\
  &\leq \frac{\|h[T_{0}]\|_{2,2}}{\sqrt{t-1}}
  +\frac{2\|x[\mathcal{I}_{0}^{c}]\|_{2,1}}{\sqrt{k(t-1)}} \notag\\
  &\leq \frac{\|h[T_{0}]+h[T_{1}]\|_{2,2}}{\sqrt{t-1}}
  +\frac{2\|x[\mathcal{I}_{0}^{c}]\|_{2,1}}{\sqrt{k(t-1)}} \notag\\
  &=\frac{X+P}{\sqrt{t-1}},
\end{align}
 where the second inequality follows from applying Cauchy-Schwarz to any block $k-$sparse vector $\upsilon$, $\|\upsilon\|_{2,1}=\sum_{i}\|\upsilon[i]\|_{2}\cdot 1\leq \sqrt{k} \|\upsilon\|_{2,2}$.
Take
$\beta_{i}=h[T_{0}]+h[T_{1}]+\mu u_{i},$
where $0\leq \mu \leq 1$.
We observe that
\begin{align}\label{g9}
  \sum\limits_{j=1}^{J}\lambda_{j}\beta_{j}-\frac{1}{2}\beta_{i}
  &=h[T_{0}]+h[T_{1}]+\mu h[T_{2}]-\frac{1}{2}\beta_{i} \notag\\
  &=(\frac{1}{2}-\mu)(h[T_{0}]+h[T_{1}])+\mu h-\frac{\mu}{2} u_{i}.
\end{align}
Moreover, $\beta_{i}$ and  $\sum\limits_{j=1}^{N}\lambda_{j}\beta_{j}-\frac{1}{2}\beta_{i}-\mu h$
are block $tk-$sparse,
because $h[T_{0}]$ is block $k-$sparse, $h[T_{1}]$ is block $l-$sparse,
 and $u_{i}$ is block $k(t-1)-l-$sparse.
Note that the following identity (see (25) in \cite{CZ})
\begin{align}\label{g10}
 \sum\limits_{i=1}^{J}\lambda_{i}
\Big\|A\Big(\sum\limits_{j=1}^{J}\lambda_{j}\beta_{j}
 -\frac{1}{2}\beta_{i}\Big)\Big\|_{2}^{2}
=\sum\limits_{i=1}^{J}\frac{\lambda_{i}}{4}\|A\beta_{i}\|_{2}^{2}.
\end{align}
We first bound the left hand side of \eqref{g10}.
Substituting \eqref{g9} into the left hand side of \eqref{g10}
and combining \eqref{g7} and the definition of block RIP, we have the upper bound
\begin{align*}
\sum\limits_{i=1}^{J}\lambda_{i}
&\Big\|A\Big(\sum\limits_{j=1}^{J}\lambda_{j}\beta_{j}
-\frac{1}{2}\beta_{i}\Big)\Big\|_{2}^{2} \\
         =&\sum\limits_{i=1}^{J}\lambda_{i}
\Big\|A\Big((\frac{1}{2}-\mu)(h[T_{0}]+h[T_{1}])+\mu h-\frac{\mu}{2} u_{i}\Big)\Big\|_{2}^{2} \\
         =&\sum\limits_{i=1}^{J}\lambda_{i}\Big \|A
\Big((\frac{1}{2}-\mu)(h[T_{0}]+h[T_{1}])-\frac{\mu}{2} u_{i}\Big)\Big\|_{2}^{2}  \\
         &+2\left\langle A\left( (\frac{1}{2}-\mu)(h[T_{0}]+h[T_{1}])
-\frac{\mu}{2} h[T_{2}]\right), \mu Ah\right\rangle+\mu^{2}\|Ah\|_{2}^{2} \\
        =&\sum\limits_{i=1}^{J}\lambda_{i}\Big \|A
\Big((\frac{1}{2}-\mu)(h[T_{0}]+h[T_{1}])-\frac{\mu}{2} u_{i}\Big)\Big\|_{2}^{2}  \\
        &+\mu(1-\mu)\langle A(h[T_{0}]+h[T_{1}]), Ah\rangle  \\
\leq&(1+\delta_{tk})\sum\limits_{i=1}^{J}\lambda_{i}
\Big\|(\frac{1}{2}-\mu)(h[T_{0}]+h[T_{1}])-\frac{\mu}{2}u_{i}\Big\|_{2}^{2}  \\
       &+\mu(1-\mu)\sqrt{1+\delta_{tk}}\|h[T_{0}]+h[T_{1}]\|_{2}
\cdot(2\varepsilon) \\
       =&(1+\delta_{tk})\sum\limits_{i=1}^{J}\lambda_{i}
\Big((\frac{1}{2}-\mu)^{2}\|h[T_{0}]+h[T_{1}]\|_{2}^{2}
+\frac{\mu^{2}}{4} \|u_{i}\|_{2}^{2}\Big) \\
       &+\mu(1-\mu)\sqrt{1+\delta_{tk}}\|h[T_{0}]+h[T_{1}]\|_{2}
\cdot(2\varepsilon).
 \end{align*}
On the other hand, using the expression of $\beta_{i}$, the block RIP results in the lower bound
\begin{align*}
 \sum\limits_{i=1}^{J}\frac{\lambda_{i}}{4}\|A\beta_{i}\|_{2}^{2}&= \sum\limits_{i=1}^{J}\frac{\lambda_{i}}{4}\|A\Big(h[T_{0}]+h[T_{1}]+\mu u_{i}\Big)\|_{2}^{2} \\
&\geq \sum\limits_{i=1}^{J}\frac{\lambda_{i}}{4}(1-\delta_{tk})
\|h[T_{0}]+h[T_{1}]+\mu u_{i}\|_{2}^{2} \\
&=(1-\delta_{tk})\sum\limits_{i=1}^{J}\frac{\lambda_{i}}{4}\Big(\|h[T_{0}]+h[T_{1}]\|_{2}^{2}+\mu^{2}\|u_{i}\|_{2}^{2}\Big).
\end{align*}
Combining the above two inequalities, we have
\begin{align*}
  \Big[&(\mu^{2}-\mu)+(\frac{1}{2}-\mu+\mu^{2})\delta_{tk}\Big]\|h[T_{0}]+h[T_{1}]\|_{2}^{2} \\
  &+\mu(1-\mu)\sqrt{1+\delta_{tk}}\|h[T_{0}]+h[T_{1}]\|_{2}\cdot(2\varepsilon)
  +\sum\limits_{i=1}^{J}\lambda_{i}\frac{\delta_{tk}}{2}\mu^{2}\|u_{i}\|_{2}^{2}\geq 0.
\end{align*}
From (\ref{g8}) and the  expression of $X$ with the fact $\|\cdot\|_{2,2}=\|\cdot\|_{2}$ , we obtain
\begin{align*}
 \Big[&(\mu^{2}-\mu)+\Big(\frac{1}{2}-\mu+(1+\frac{1}{2(t-1)})\mu^{2}\Big)\delta_{tk}\Big]X^{2} \\
  &+\Big[2\varepsilon\mu(1-\mu)\sqrt{1+\delta_{tk}}
  +\frac{\mu^{2}\delta_{tk}P}{t-1}\Big]X
  +\frac{\mu^{2}P^{2}\delta_{tk}}{2(t-1)}\geq 0.
 \end{align*}
Taking $\mu=\sqrt{t(t-1)}-(t-1)$, we have
$$\frac{\mu^{2}}{t-1}\Big[-t\Big(\sqrt{\frac{t-1}{t}}-\delta_{tk}\Big)X^{2}
+\Big(2\varepsilon\sqrt{t(t-1)(1+\delta_{tk})}+P\delta_{tk}\Big)X
+\frac{P^{2}\delta_{tk}}{2}\Big]\geq 0.$$
The condition $\delta_{tk}<\sqrt{(t-1)/t}$ ensures that the above inequality is a second-order inequality for $X$ and the quadratic coefficient is negative.
Thus, we obtain
\begin{align*}
 X\leq&\Big\{\Big( 2\varepsilon\sqrt{t(t-1)(1+\delta_{tk})}+P\delta_{tk} \Big)\\
 &+\Big[\Big( 2\varepsilon\sqrt{t(t-1)(1+\delta_{tk})}+P\delta_{tk} \Big)^{2}\\
 &+2t\Big( \sqrt{(t-1)/t}-\delta_{tk} \Big)P^{2}\delta_{tk}  \Big]^{1/2}\Big\}
\cdot \Big(2t(\sqrt{(t-1)/t}-\delta_{tk})\Big)^{-1}\\
\leq&\frac{2\sqrt{t(t-1)(1+\delta_{tk})}}{t(\sqrt{(t-1)/t}-\delta_{tk})}\varepsilon \\
&+\frac{2\delta_{tk}+\sqrt{2t(\sqrt{(t-1)/t}
-\delta_{tk})\delta_{tk}}}{2t(\sqrt{(t-1)/t}-\delta_{tk})}P,
\end{align*}
where the last inequality is a result of the fact that
$(a+b)^{\frac{1}{2}}\leq a^{\frac{1}{2}}+b^{\frac{1}{2}}$ for any nonnegative constants $a$ and $b$.
With \eqref{g3} and the representation of $P$,  it is clear that
$$\|h[T_{0}^{c}]\|_{2,1}\leq \|h[T_{0}]\|_{2,1}+P\sqrt{k}.$$
From Lemma \ref{l2}, we conclude that
$$\|h[T_{0}^{c}]\|_{2,2}\leq \|h[T_{0}]\|_{2,2}+P.$$
Therefore, it is not hard to see that
\begin{align*}
\|h\|_{2}&=\sqrt{\|h[T_{0}]\|^{2}_{2}+\|h[T_{0}^{c}]\|^{2}_{2}}  \\
&\leq\sqrt{\|h[T_{0}]\|^{2}_{2}+(\|h[T_{0}]\|_{2}+P)^{2}} \\
&\leq\sqrt{2}\|h[T_{0}]\|_{2}+P \\
&\leq\sqrt{2}(\|h[T_{0}]+h[T_{1}]\|_{2})+P \\
&=\sqrt{2}X+P \\
&\leq \frac{2\sqrt{2t(t-1)(1+\delta_{tk})}}{t(\sqrt{(t-1)/t}-\delta_{tk})}\varepsilon \\
&+\Big(\frac{\sqrt{2}\delta_{tk}+\sqrt{t(\sqrt{(t-1)/t}-\delta_{tk})\delta_{tk}}}{t(\sqrt{(t-1)/t}-\delta_{tk})}+1\Big)
\frac{2\|x[\mathcal{I}_{0}^{c}]\|_{2,1}}{\sqrt{k}}.
\end{align*}

If $tk$ is not an integer, we denote $t'=\lceil tk\rceil/k$, then $t'k$ is an integer and $t< t'$.
So we have
$\delta_{t'k}=\delta_{tk}<\sqrt{\frac{t-1}{t}}<\sqrt{\frac{t'-1}{t'}}.$
Using the method similar to the proof above,
we can prove the result by working on $\delta_{t'k}$.
Hence, we complete the proof of the theorem.



\end{proof}

\begin{Remark}
Theorem \ref{t1} indicates that as long as measurement matrix $A$
meets with the block RIP with a suitable constant,
the mixed $l_{2}/l_{1}$ minimization method can robustly recover any signals with block structure from noisy measurements $y=Ax+z$.
Moreover, if $x$ is block $k-$sparse, then Theorem \ref{t1} guarantees perfect and stable recovery of $x$ from its samples $y$ in the noise-free and noisy setting.

\end{Remark}

The following theorem shows the condition $\delta_{tk}<\sqrt{\frac{t-1}{t}}$ with $t\geq4/3$ is sharp for exact and stable recovery in noiseless and noise case, respectively.
\begin{Theorem}
Let $t\geq \frac{4}{3}$.
For any $\varepsilon>0$ and $k\geq \frac{5}{\varepsilon}$,
then there exist a sensing matrix $A\in \mathbb{R}^{n\times N}$ with $\delta_{tk}<\sqrt{\frac{t-1}{t}}+\varepsilon$
 and some block $k-$sparse signal $x_{0}$ such that
 \begin{description}
   \item[{\rm (1)}]In the noiseless case, i.e., $y=Ax_{0}$,
   the mixed $l_{2}/l_{1}$ minimization \eqref{g11} can not exactly recover the block $k-$sparse signal $x_{0}$, i.e., $\widehat{x}\neq x_{0}$,
   where $\widehat{x}$ is the solution to \eqref{g11}.
   \item[{\rm (2)}] In the noise case, i.e., $y=Ax_{0}+z$,
   the mixed $l_{2}/l_{1}$ minimization \eqref{g11} can not stably recover the block $k-$sparse signal $x_{0}$, i.e., $\widehat{x}\nrightarrow x_{0}$ as $z\rightarrow 0$,
   where $\widehat{x}$ is the solution to \eqref{g11}.
 \end{description}
\end{Theorem}

\begin{proof}
For all $\varepsilon>0$ and $k\geq \frac{5}{\varepsilon}$, let
$a'=((t-1)+\sqrt{t(t-1)})k$ and $N\geq (k+a')d$. Since $t\geq \frac{4}{3}$,
 we have $a'\geq k$. Suppose that $a$ is the largest integer strictly smaller than $a'$, then $a<a'$ and $a'-a\leq 1$.
 Denote
 $$\gamma=\frac{1}{\sqrt{kd+\frac{ak^{2}}{a'^{2}}d}}
 (\underbrace{\overbrace{1, \ldots, 1}^{d},
 \ldots,\overbrace{1, \ldots, 1}^{d}}_{k~ blocks},
 \underbrace{\overbrace{-\frac{k}{a'},\ldots,-\frac{k}{a'}}^{d},
 \ldots,\overbrace{-\frac{k}{a'},\ldots,-\frac{k}{a'}}^{d}}_{a~ blocks},
 \textbf{0},\ldots,\textbf{0})\in \mathbb{R}^{N}
 ,$$ where $k+a\leq M$, $\mathcal{I}=\{d_{1}=d,\ldots,d_{k+a}=d, d_{k+a+1},\ldots,d_{M}\}$ and $\|\gamma\|_{2}=1$.
Define the linear map $A: \mathbb{R}^{N}\rightarrow\mathbb{R}^{N}$ by
\begin{align*}
 Ax&=\sqrt{1+\sqrt{\frac{t-1}{t}}}(x-\langle \gamma, x\rangle \gamma),
\end{align*}
for all $x\in\mathbb{R}^{N}.$
Then for any block $\lceil tk\rceil-$sparse signal $x$, we have
$$\|Ax\|_{2}^{2} =\left(1+\sqrt{\frac{t-1}{t}}\right)
\left(\|x\|_{2}^{2}-|\langle \gamma, x\rangle|^{2}\right).$$
From the Cauchy-Schwarz inequality and $a'\geq k$, $a'-a\leq 1$
as well as the fact
$\frac{a'^{2}+k^{2}(t-1)}{a'^{2}+a'k}
 =2\sqrt{t-1}(\sqrt{t}-\sqrt{t-1})$ (see the proof of Theorem 2.2, \cite{CZ}),
therefore,
\begin{align*}
 0\leq&|\langle \gamma, x\rangle|^{2}
 =|\langle \gamma[\mathrm{supp}[x]], x\rangle|^{2}
 \leq\|x\|_{2}^{2}\cdot \|\gamma[\max{(\lceil tk\rceil)}]\|_{2}^{2}  \\
\leq&\|x\|_{2}^{2}\cdot\frac{a'^{2}+k(\lceil tk\rceil-k)}{a'^{2}+ak}
\leq\frac{a'^{2}+k^{2}(t-1)+k}{a'^{2}+ak}\|x\|_{2}^{2}\\
=&\frac{a'^{2}+k^{2}(t-1)+k}{a'^{2}+a'k}
\cdot\frac{a'^{2}+a'k}{a'^{2}+ak}\|x\|_{2}^{2}\\
=&\frac{a'^{2}+k^{2}(t-1)+k}{a'^{2}+a'k}
\cdot\frac{1}{1-\frac{k(a'-a)}{a'^{2}+a'k}}\|x\|_{2}^{2}\\
=&\frac{a'^{2}+k^{2}(t-1)}{a'^{2}+a'k}
\cdot\frac{a'^{2}+k^{2}(t-1)+k}{a'^{2}+k^{2}(t-1)}
\cdot\frac{1}{1-\frac{k(a'-a)}{a'^{2}+a'k}}\|x\|_{2}^{2}\\
\leq&2\sqrt{t-1}(\sqrt{t}-\sqrt{t-1})
\cdot(1+\frac{1}{tk})\cdot\frac{1}{1-\frac{1}{2k}}\|x\|_{2}^{2}\\
\leq&2\left(\sqrt{t(t-1)}-(t-1)\right)
\cdot(1+\frac{5}{2k})\|x\|_{2}^{2}\\
\leq&\left(2\sqrt{t(t-1)}-2(t-1)+\frac{5}{2k}\right)\|x\|_{2}^{2},
\end{align*}
where $\gamma[\max{(\lceil tk\rceil)}]$ denotes a vector
remaining the $\lceil tk\rceil$ blocks with largest $l_{2}$
norm of $\gamma$ and zero elsewhere
and $\|\gamma[\max{(\lceil tk\rceil)}]\|_{2}^{2}
\leq \frac{a'^{2}+k(\lceil tk\rceil-k)}{a'^{2}+ak}$.
The last inequality follows that
$2\sqrt{t(t-1)}-2(t-1)\leq 1$.
Consequently,
\begin{align*}
      \left(1+\sqrt{\frac{t-1}{t}}+\varepsilon\right)\|x\|_{2}^{2}
      \geq&\left(1+\sqrt{\frac{t-1}{t}}\right)\|x\|_{2}^{2}\geq\|Ax\|_{2}^{2}\\
      \geq&\left(1+\sqrt{\frac{t-1}{t}}\right)
      \left(1-2\sqrt{t(t-1)}+2(t-1)-\frac{5}{2k}\right)\|x\|_{2}^{2}\\
      =&\Bigg[\left(1+\sqrt{\frac{t-1}{t}}\right)
      \left(1-2\sqrt{t(t-1)}+2(t-1)\right)\\
      &-\left(1+\sqrt{\frac{t-1}{t}}\right)\frac{5}{2k} \Bigg]\|x\|_{2}^{2}\\
      =&\left[1-\sqrt{\frac{t-1}{t}}-
      \left(1+\sqrt{\frac{t-1}{t}}\right)\frac{5}{2k} \right]\|x\|_{2}^{2}\\
      \geq&\left(1-\sqrt{\frac{t-1}{t}}-\varepsilon\right)\|x\|_{2}^{2},
     \end{align*}
where the last inequality follows from that $1+\sqrt{\frac{t-1}{t}}\leq 2$ and $k\geq \frac{5}{\varepsilon}$.
It follows immediately that $\delta_{tk}\leq \sqrt{\frac{t-1}{t}}+\varepsilon$.
Next, we define
\begin{align*}
x_{0}&=
 (\underbrace{\overbrace{1, \ldots, 1}^{d},
 \ldots,\overbrace{1, \ldots, 1}^{d}}_{k~ blocks},
 \underbrace{\overbrace{0,\ldots,0}^{d},
 \ldots,\overbrace{0,\ldots,0}^{d}}_{a~ blocks},
 \textbf{0},\ldots,\textbf{0})\in \mathbb{R}^{N},  \\
\gamma_{0}&=
 (\underbrace{\overbrace{0, \ldots, 0}^{d},
 \ldots,\overbrace{0, \ldots, 0}^{d}}_{k~ blocks},
 \underbrace{\overbrace{\frac{k}{a'},\ldots,\frac{k}{a'}}^{d},
 \ldots,\overbrace{\frac{k}{a'},\ldots,\frac{k}{a'}}^{d}}_{a~ blocks},
 \textbf{0},\ldots,\textbf{0})\in \mathbb{R}^{N}.
 \end{align*}
Note that $\|x_{0}\|_{2,1}=k\sqrt{d}$, $\|\gamma_{0}\|_{2,1}\leq \frac{a}{a'}\cdot k\sqrt{d}<k\sqrt{d}$ and
 $x_{0}$ is block $k-$sparse, $\gamma=\frac{1}{\sqrt{kd+\frac{ak^{2}}{a'^{2}}d}}(x_{0}-\gamma_{0})$.
Since $A\gamma=0$, we have $Ax_{0}=A\gamma_{0}$.

Thus, in the noiseless case $y=Ax_{0}$,
suppose that the mixed $l_{2}/l_{1}$ minimization method \eqref{g11} can exactly recover $x_{0}$,
 i.e., $\widehat{x}=x_{0}$. From the definition of $\widehat{x}$ and $y=A\gamma_{0}$,
 it contradicts that $\|\gamma_{0}\|_{2,1}< \|x_{0}\|_{2,1}$.

 In the noise case $y=Ax_{0}+z$, suppose that the mixed $l_{2}/l_{1}$ minimization method \eqref{g11} can stably recover $x_{0}$,
 i.e., $\lim\limits_{z\rightarrow0}\widehat{x}=x_{0}$. We observe that $y-A(\widehat{x}-x_{0}+\gamma_{0})=y-A\widehat{x}$,
thus $\|\widehat{x}\|_{2,1}\leq \|\widehat{x}-x_{0}+\gamma_{0}\|_{2,1}$. So $\|x_{0}\|_{2,1}\leq \|\gamma_{0}\|_{2,1}$ as $z\rightarrow 0$.
It contradicts that $\|\gamma_{0}\|_{2,1}< \|x_{0}\|_{2,1}$.

Therefore, the mixed $l_{2}/l_{1}$ minimization method \eqref{g11} fails to  exactly and stably recover $x_{0}$ based on $y$ and $A$.
\end{proof}

\section{Conclusion}\label{4}
In this paper, we consider the problem of recovering for an unknown
signal with additional structure-its entries are not
dispersing over all the signal vector but arising in clusters (or blocks)-from a given set of noisy linear measurements.
Based on block RIP, we mainly investigate the recovery guarantee
for the mixed $l_{2}/l_{1}$ minimization method. By extending the technique
 of sparse representation of a polytope to the block setting, we
 establish a high order block RIP condition for robust recovery of signals
 with block pattern by mixed $l_{2}/l_{1}$ minimization in the presence of noise. We also proved under the same condition, the block sparse signals
 can be exactly and stably recovered in the noiseless and noise case, respectively.
 Furthermore, in order to clarify its optimality, we also give a specific
 measurement matrix and block sparse signal such that the given concrete signal can not be exactly and stably recovered from its samples via  mixed $l_{2}/l_{1}$ minimization, when the upper bound of $\delta_{tk}$ increases any $\varepsilon$.
Also, if $d_{i}=1$ for $i\in\{1,\ldots,M\}$, our results return to
those of Cai and Zhang \cite{CZ}.
\section*{Acknowledgments}
This work was supported by the NSF of China (Nos.11271050, 11371183).


\begin{thebibliography}{12}
\bibitem{AT}
Akcakaya, M., Tarokh, V.: A frame construction and a universal distortion bound for sparse representations. IEEE Trans. Signal Process. 56(6), 2443-2450 (2008)

\bibitem{B}
 Bach, F. R.: Consistency of the group Lasso and multiple kernel
learning. J. Mach. Learn. Res. 9, 1179-1225 (2008)

\bibitem{BDDW}
Baraniuk, R., Davenport, M., DeVore, R., Wakin, M.: A simple proof of the restricted isometry property for random matrices. Constr. Approx. 28(3), 253-263 (2008)

\bibitem{BW}
Bechler, P., Wojtaszczyk, P.: Error estimates for orthogonal matching pursuit and random dictionaries. Constr. Approx. 33(2), 273-288 (2011)


\bibitem{CWX}
Cai, T. T., Wang, L., Xu, G. W.: New bounds for restricted isometry constants. IEEE Trans. Inform. Theory 56(9), 4388-4394 (2010)

\bibitem{CWX1}
Cai, T. T., Wang, L., Xu, G. W.: Shifting inequality and recovery of sparse signals. IEEE Trans. Signal Process. 58(3), 1300-1308 (2010)

\bibitem{CXZ}
Cai, T. T., Xu, G. W., Zhang, J.: On recovery of sparse signal via $l_{1}$ minimization. IEEE Trans. Inform. Theory 55(7), 3388-3397 (2009)

\bibitem{CZ}
Cai, T. T., Zhang, A.: Spares representation of a polytope and recovery of sparse signals and low-rank matrices. IEEE Trans. Inform. Theory 60(1), 122-132 (2014)

\bibitem{CZ1}
Cai, T. T., Zhang, A.: Sharp RIP bound for sparse signal and low-rank matrix recovery.  Appl. Comput. Harmon. Anal. 35, 74-93 (2013)

\bibitem{CZ2}
Cai, T. T., Zhang, A.: Compressed sensing and affine rank minimization under restricted isometry. IEEE Trans. Signal Process. 61(13), 3279-3290 (2013)

\bibitem{C1}
Cand\`{e}s, E. J.: The restricted isometry property and its
implications for compressed sensing. Compte Rendus Math. 346(9-10),
589-592 (2008)
\bibitem{CRT}
Cand\`{e}s, E. J., Romberg, J., Tao, T.: Stable signal recovery from incomplete and inaccurate measurements. Commum. Pure Appl. Math. 59, 1207-1223 (2006)

\bibitem{CT}
Cand\`{e}s, E. J., Tao, T.: Decoding by linear programming. IEEE Trans. Inform. Theory 51(12), 4203-4215 (2005)

\bibitem{CH}
 Chesneau, C., Hebiri, M.: Some theoretical results on the Grouped
Variables Lasso. Math. Methods Statist. 17(4),  317-326
(2008)

\bibitem{CDD}
Cohen, A., Dahmen, W., DeVore, R.: Comopressed sensing and best $k$-term
approximation. J. Am. Math. Soc. 22(1), 211-231 (2009)

\bibitem{CDD1}
Cohen, A., Dahmen, W., DeVore, R.: Instance optimal decoding by thresholding in compressed sensing. in Proc. 8th Int. Int. Conf., Jun. 16-20, 2010, El Escorial,
Madrid, Spain, 505, 1 (2010), Amer. Mathematical Society

\bibitem{CDD2}
 Cohen, A., Dahmen, W., DeVore, R.: Orthogonal matching pursuit under the restricted isometry property. Constr. Approx. doi:10.1007/s00365-016-9338-2 (May 2016)

\bibitem{CR}
Cotter, S., Rao, B.: Sparse channel estimation via matching pursuit
with application to equalization. IEEE Trans. Commun. 50(3),
374-377 (2002)

\bibitem{EKB}
Eldar, Y. C., Kuppinger, P., Bolcskei, H.: Block-sparse signals: uncertainty relations and efficient recovery. IEEE Trans. Signal Process.
58(6), 3042-3054 (2010)

\bibitem{EM}
Eldar, Y. C., Mishali, M.: Robust recovery of signals from a structured
union of subspaces. IEEE Trans. Inform. Theory 55(11), 5302-5316 (2009)

\bibitem{ER}
Eldar, Y. C., Rauhut, H.: Average case analysis of multichannel
sparse recovery using convex relaxation. IEEE Trans. Inform. Theory
56(1), 505-519 (2010)

\bibitem{EV}
Elhamifar, E., Vidal, R.: Block-sparse recovery via convex optimization. IEEE Trans. Signal Process. 60(8), 4094-4107 (2012)

\bibitem{F}
Foucart, S.: A note on guaranteed sparse recovery via
$l_{1}-$minimization. Appl. Comput. Harmon. Anal. 29, 97-103 (2010)

\bibitem{F1}
Foucart, S.: Sparse recovery algorithms: sufficient conditions in
terms of restricted isometry constants. Springer Proceedings in
Mathematics  13, 65-77 (2012)

\bibitem{FLZZ}
Fu, Y. L., Li, H. F., Zhang, Q. H., Zou, J.: Block-sparse recovery via redundant block OMP. Signal Process. 97, 162-171 (2014)

\bibitem{HS}
Herman, M., Strohmer, T.: High-resoluton radar via compressed sensing. IEEE Trans. Signal Process. 57(6),  2275-2284 (2009)

\bibitem{HZ}
 Huang, J., Zhang, T.: The benefit of group sparsity. Annals of Statistics 38(4), 1978-2004 (2010)

\bibitem{LL}
Lin, J.H., Li, S.: Block sparse recovery via mixed $l_2/l_1$ minimization.
Acta Math. Sin. 29(7), 1401-1412 (2013)

\bibitem{LT}
Liu, E., Temlyakov, V. N.: The orthogonal super greedy algorithm and applications in compressed sensing. IEEE Trans. Inform. Theory 58(4), 2040-2047 (2012)


\bibitem{LT1}
Livshitz, E. D., Temlyakov, V. N.: Sparse approximation and recovery by greedy algorithms. IEEE Trans. Inform. Theory 60(7), 3989-4000 (2014)


\bibitem{LDSP}
Lustig, M., Donoho, D. L., Santos, J. M., Pauly, J. M.: Compressed sensing MRI. IEEE Signal Process. Mag. 25(2),  72-82 (2008)

\bibitem{LWB}
Lv, X. L., Wan, C. R., Bi, G. A.: Block orthogonal greedy algorithm
for stable recovery of block-sparse signal representations.
Signal Process. 90, 3265-3277 (2010)

\bibitem{MW}
Majumdar, A., Ward,R. K.: Compressen sensing of color images. Signal Process. 90(12), 3122-3127 (2010)

\bibitem{MGB}
 Meier, L., Geer, S., B\"uhlmann, P.: The group lasso for logistic
regression. J. R. Statist. Soc. B 70(1), 53-77 (2008)

\bibitem{ME}
Mishali, M., Eldar, Y. C.: Blind multi-band signal reconstruction: compressed sensing for analog signals. IEEE Trans. Signal Process. 57(3), 993-1009 (2009)

\bibitem{ME1}
Mishali, M., Eldar, Y. C.: Reduce and boost: Recovering arbitrary
sets of jointly sparse vectors. IEEE Trans. Signal Process. 56(10), 4692-4702 (2008)

\bibitem{ML}
Mo, Q., Li, S.: New bounds on the restricted isometry constant $\delta_{2k}$. Appl. Comput. Harmon. Anal. 31(3), 460-468 (2011)

\bibitem{NR}
 Nardi, Y., Rinaldo, A.: On the asymptotic properties of the group
Lasso estimator for linear models. Electron. J. Statist. 2,
605-633 (2008)

\bibitem{PVMH}
Parvaresh, F., Vikalo, H., Misra, S., Hassibi, B.: Recovering sparse signals using sparse measurement matrices in compressed DNA microarrays. IEEE J. Sel. Top. Signal Process. 2(3), 275-285 (2008)

\bibitem{R}
 Rauhut, H.: Compressive sensing and structured random matrices. in Theoretical Foundations and Numerical Methods
  for Sparse Recovery, M. Fornasier, editor, Radon Series Comp. Appl. Math. 9, de Gruyter, 9, 1-92 (2010)

\bibitem{RRT}
Rauhut, H., Romberg, J., Tropp, J.:  Restricted isometries for partial random circulant matrices. Appl. Comput. Harmom. Anal. 32(2), 242-254 (2012)


\bibitem{T1}
Temlyakov, V. N.: Greedy approximation in convex optimization. Constr. Approx. 41(2), 269-296 (2015)

\bibitem{T}
Tropp, J. A.: Greed is good: algorithmic results for sparse approximation. IEEE Trans. Inform. Theory 50(10), 2231-2242 (2004)

\bibitem{TG}
Tropp, J. A., Gilbert, A.: Signal recovery from partial information via orthogonal matching pursuit. IEEE Trans. Inform. Theory 53(12), 4655-4666 (2007)


\bibitem{WBWL}
Wang, L., Bi, G. A., Wan, C. R., Lv, X. L.: Improved stability conditions of BOGA
 for noisy block-sparse signals. Signal Process. 91, 2567-2574 (2011)

\bibitem{WWX}
Wang, Y., Wang, J. J., Xu, Z. B.: Restricted $p-$isometry properties of nonconvex
block-sparse compressed sensing. Signal Process. 104, 188-196 (2014)

\bibitem{WZLLT}
Wen, J. M., Zhou, Z. C., Liu, Z. L., Lai, M. J., Tang, X. H.:
Sharp sufficient conditions for stable recovery of block sparse signals by block orthogonal matching pursuit. arXiv:1605.02894v1 [cs.IT] (May 2016)


\bibitem{YL}
Yuan, M., Lin, Y.: Model selection and estimation in regression
with grouped variables.  J. R. Statist. Soc. B
68(1), 49-67 (2006)

\bibitem{ZZZ}
Zhang, J., Zhu, D., Zhang, G.: Adaptive compressed sensing radar oriented toward cognitive detection in dynamic sparse target scene. IEEE Trans. Signal Process. 60(4), 1718-1729 (2012)

\end{thebibliography}
\end{document}